\documentclass[11pt]{article}

\usepackage{algorithm}
\usepackage{algorithmic}
\usepackage{tikz}
\usetikzlibrary{arrows}
\usepackage{caption}
\usepackage{subcaption}
\usepackage{color}

\usepackage{amssymb}
\usepackage{amsmath}
\usepackage{fourier}
\usepackage{fullpage}

\usepackage{graphicx}
\usepackage{epigraph}
\usepackage{amssymb}
\newcommand{\bi}{\begin{itemize}}
\newcommand{\ei}{\end{itemize}}
\newcommand{\be}{\begin{enumerate}}
\newcommand{\ee}{\end{enumerate}}

\newtheorem{theorem}{Theorem}[section]
\newtheorem{lemma}[theorem]{Lemma}
\newtheorem{proposition}[theorem]{Proposition}
\newtheorem{corollary}[theorem]{Corollary}

\newenvironment{proof}[1][Proof:]{\begin{trivlist}
\item[\hskip \labelsep {\bfseries #1}]}{\end{trivlist}}

\newcommand{\qed}{\hfill\raisebox{1mm}{\framebox[.6em][l]{\rule{0mm}{.3mm}}}\vspace{2mm}}

\begin{document}
\title{Weakly Submodular Functions}
\author{
Allan Borodin\thanks{
Dept of Computer Science, University of Toronto, bor@cs.toronto.edu} \and
Dai Tri Man L\^e\thanks{ Altera Corporation; work done while 
at the Department of Computer Science, University of Toronto, ledt@cs.toronto.edu}\and
Yuli Ye\thanks{
Amazon Inc., yeyuli@amazon.com}
}

\date{\today}

\maketitle
\begin{abstract}
Submodular functions are well-studied in combinatorial optimization, game theory and economics. The natural diminishing returns property makes them suitable for many applications.
We study an extension of monotone submodular functions, which we call {\em weakly submodular functions}. Our extension 
includes
some (mildly) supermodular functions. We show that 
several natural functions belong to this class and relate our class to 
some other recent submodular function extensions. 

We consider the optimization problem of maximizing a weakly submodular function subject to uniform and general matroid constraints. For 
a uniform matroid constraint, the ``standard greedy algorithm'' achieves a 
constant approximation ratio where the constant 
(experimentally) converges to 5.95 as the 
cardinality constraint increases. For a general matroid constraint, a simple local search algorithm achieves a
constant approximation ratio where the constant (analytically) 
converges to 10.22 as
the rank of the matroid increases.
\end{abstract}

{\bf keywords:} submodular functions, max-sum dispersion, greedy algorithms, local 
search

\newpage 

\section{Introduction}

\label{section:wsubmodular}

There are many applications where the goal becomes a problem of maximizing
a submodular function subject to some constraint. In many cases the
submodular function $f$ is also monotone, non-negative and normalized so that
$f(\emptyset) = 0$. 
Such applications arise for example in the consideration
of influence in a stochastic social network as formalized in Kempe, Kleinberg
and Tardos \cite{KKT03}, diversified search ranking as in Bansal, Jain, Kazeykina and Naor  \cite{BJKN10} 
and document summarization as in Lin and Bilmes \cite{LB11}.  
In another application, following the work of Gollapudi and Sharma \cite{GS09}, 
Borodin, Lee and Ye \cite{BLY12}
considered the linear combination of a  
monotone submodular function that measures
the ``quality'' of a set of results combined with a diversity function 
given by the max-sum dispersion measure, a widely studied measure 
of diversity. Their analysis suggested that although the max-sum dispersion
measure is a supermodular function, it possessed similar properties 
to monotone submodular functions. In this paper we develop this idea by
introducing the class of \emph{weakly submodular} functions and show that 
greedy and local search algorithms can be used (respectively) to 
approximately maximize
such functions subject to a cardinality (resp. general matroid) constraint. 
More specifically, for any cardinality constraint $p$, the greedy
algorithm has a constant approximation ratio $\alpha(p)$ that 
experimentally appears to be converging (from below) to $5.95$ 
as $p$ increases. For a general matroid constraint with 
rank $s$, we prove that 
the local search algorithm has constant approximation ratio $\rho(s)$ 
which analytically is converging (from above) to $10.22$ as $s$ 
increases.  

The literature on the maximization of submodular functions
is extensive. Here we only mention the most relevant work.
Perhaps the most seminal paper concerning monotone submodular functions
is the Nemhauser, Fisher and Wosley paper \cite{NWF78} showing that natural
greedy and local search algorithms for maximizing a monotone submodular
function obtains approximation ratios $\frac{e}{e-1}$ (resp. 2)
for maximizing any monotone submodular function subject to a
cardinality (resp. arbitrary matroid) constraint. Our work shows that these
algorithms still enjoy constant approximation ratios for the broader class
of weakly submodular functions. More recent
work (see, \cite{FeigeMV2007,BuchbinderFNS12,BuchbinderFNS14}) provides constant approximation bounds for
unconstrained and constrained non monotone submodular functions.
 
The remainder of the paper is as follows. In section \ref{sec:prelim}, we
provide the definition of weakly submodular 
\footnote{Unfortunately, the term ``weakly submodular'' has been used before in
the context of lattices by Wild \cite{Wild08}. In that usage, such functions
are a subclass of submodular functions rather than a larger class. It
is difficult to find an appropriate name for class of functions
studied in this paper.  For example, we would have preferred  to have used
the term
{\it meta submodular} but that term is already used in the
computer science community
\cite{KPR09}. We do not believe that this abuse of terminology will
cause any confusion within the computer science community.} functions. 
In section \ref{sec:examples} we provide some basic observations
about this class of functions along with a number of examples of monotone 
submodular function (that are not submodular). Section \ref{sec:related-work} 
contains a discussion of two other frameworks for extending 
submodular functions.  Sections \ref{sec:cardinality} and \ref{ssec:wsub-fd}
contain analyses of the approximation ratios 
of the natural greedy (respectively local search) algorithms 
for maximizing monotone weakly submodular functions subject to cardinality
(respectively, matroid) constraints. We conclude in section \ref{sec:conclusion}
with some open problems.       

\section{Preliminaries}

\label{sec:prelim}

Let $f: U \rightarrow \Re$ be a set function over a universe $U$. 
We will restrict attention to
set functions that are 
normalized and non-negative.
That is, $f$ satifies: 
\begin{itemize}
\item
$f(\emptyset) = 0$ 
\item
$f(S) \geq 0$ for all $S \subseteq U$
\end{itemize}
For the most part, we will focus attention on functions that are monotone.
That is, 
\begin{itemize}
\item
$ f(S) \leq f(T)$ for all $S \subseteq T \subseteq U$ \\ 
\end{itemize}
  
A function $f(\cdot)$ is submodular if for any two sets $S$ and $T$, we have
$$f(S)+f(T)\ge f(S\cup T)+f(S\cap T).$$ 

We define the following generalization. We call a normalized, non-negatuve 
function $f(\cdot)$ {\em weakly submodular} 
if for any two sets $S$ and $T$, we have $$|T|f(S)+|S|f(T)\ge|S\cap T|f(S\cup T)+|S\cup T|f(S\cap T).$$

Our extension of submodularity ``normalizes'' the submodularity 
definition in terms of the cardinality of the sets occuring in the definition. 
This allows for some supermodular functions since now large sets with small
intersections can observe the required inequality. 
We will see that this generalization still retains  
the main algorithmic property of submodular functions; namely
that simple and efficient greedy and local search algorithms suffice 
to approximately maximize such functions subject to cardinality and general 
matroid constraints.

\section{Examples of Weakly Submodular Functions}

\label{sec:examples}

In this section, we will first consider some natural weakly submodular functions
showing in particular that this class includes all monotone 
submodular functions as well as some supermodular functions. 
In Section~\ref{sec:related-work}, we will relate weak submodularity  to
the functions of supermodular degree defined by Feige and Izsak 
\cite{FeigeI13} and further studied in Feldman and Izsak \cite{FeldmanI14}, 
and to the  $k$-wise dependent  
functions of Conitzer, Sandholm and Santi  \cite{ConitzerSS05}, and the 
related 
MPH-$k$ functions defined by Feige et al \cite{FeigeFIILS14}.  


\subsection{Submodular Functions}

From the weakly submodular definition, it is not obvious that 
monotone submodular functions are a subclass of weakly submodular functions.
We will prove that this is indeed the case. 

\begin{proposition}
\label{prop:subclass1}
Any monotone submodular function is weakly submodular. This, of course, implies
that every linear function (with non-negative weights) is weakly submodular. 
\end{proposition}
\begin{proof}

Given a monotone submodular function $f(\cdot)$
and two subsets $S$ and $T$, without loss of generality, we assume $|S|\le|T|$, then
$$|T|f(S)+|S|f(T)=|S|[f(S)+f(T)]+(|T|-|S|)f(S).$$
By submodularity $f(S)+f(T)\ge f(T\cup S)+f(T\cap S)$ and monotonicity $f(S)\ge f(S\cap T)$, 
we have
\begin{eqnarray*}
|T|f(S)+|S|f(T)&=&|S|[f(S)+f(T)]+(|T|-|S|)f(S)\\ 
&\ge&|S|[f(S\cup T)+f(S\cap T)]+(|T|-|S|)f(S\cap T)\\
&=&|S|f(S\cup T)+|T|f(S\cap T)\\
&=&|S\cap T|f(S\cup T)+\bigl[(|S|-|S\cap T|)f(S\cup T)+|T|f(S\cap T)\bigr].
\end{eqnarray*}
And again by monotonicity $f(S\cup T)\ge f(S\cap T)$, we have 
$$(|S|-|S\cap T|)f(S\cup T)+|T|f(S\cap T)\ge(|S|+|T|-|S\cap T|)f(S\cap T)=|S\cup T|f(S\cap T).$$
Therefore
$$|T|f(S)+|S|f(T)\ge|S\cap T|f(S\cup T)+|S\cup T|f(S\cap T);$$
the proposition follows.
\qed
\end{proof}

We note that the proof of Proposition \ref{prop:subclass1} did not
require the function $f(\cdot)$ to be normalized or non-negative. But the 
proof did use the monotonicity of $f(\cdot)$. 
Non-monotone submodular functions (such as Max-Cut and Max-Di-Cut) 
are, of course, also widely studied.
In contrast to
Proposition \ref{prop:subclass1}, if we extend the weakly submodular
definition to non-monotone functions, then it is no longer the case that
a non-monotone submodular function would necessarily be a
non-monotone weakly submodular function.

\begin{proposition}
There is a non-monotone submodular function $f(\cdot)$ that is not weakly 
submodular. More specifically, the Max-Cut function (for a 
particular graph $G$) is not weakly submodular. 
\end{proposition}

\begin{proof}

Consider a graph $G = (U,E)$ where $V = R \cup \{s\} \cup \{t\}$ and 
$E = \{(s,u), (u,t) | u \in R\}$.
Letting $S = R \cup \{s\}$ and $T = R \cup \{t\}$, for $|R| = n$ we have the following: 
\begin{itemize}
\item
$f(S) = f(T) = n$
\item
$f(S \cup T) = f(U) = 0$
\item
$f(S \cap T) = f(R) = 2n$
\end{itemize}

Thus 
\begin{enumerate}
\item $|T| f(S) + |S| f(T) = (n+1) n  + (n+1) n = 2n^2 + 2n$
\item
$|S \cap T| f(S \cup T) + |S \cup T| f(S \cap T) = 
n \cdot  0 + (n+2) \cdot 2n = 2n^2 + 4n$
\end{enumerate}

This contradicts the weakly submodular definition. 
\qed
\end{proof}

\begin{proposition}
\label{prop:complement}
Let $f$ be a monotone weakly submodular function. Then 
the complement function ${\bar f} = f(U/S)$ is 
weakly submodular iff $f$ is submodular. 

\end{proposition}

\begin{proof}

It is well known that submodular functions are closed under complememts 
so one direction of the proposition holds. We now show that when 
${\bar f}$ is also weakly submodular, then $f$ is submodular. 
(This direction holds whether or not $f$ is monotone.)

\end{proof}

On the other hand, it is easy to construct non-monontone
weakly submodular functions from any 
monotone weakly submodular function $f$ having at least one posiitve 
valuation. 
Namely, let $f(S^*) > 0$ for some $S$ with 
$\emptyset \subset S^* \subset U$. 
Then define the function 
$g$ to be identical to $f$ except that $g(U) = 0$. Clearly, $g$ is 
non-monotone. We can verify that $g$ is weakly submodular by
checking the cases where $U$ appears in the inequality that
defines weak submodularity, namely when either $S$ or $T$ is $U$, or
when $S \cup T = U$. Furthermore, if $f$ was say the metric dispersion
function, $g$ is then clearly not submodular. 

\vspace{.2in}
Hereafter, we will we restrict attention to monotone (and by definition, 
non-negative and  
normalized) 
functions. In the remaining subsections of section \ref{sec:examples}, we 
present a number of monotone submodular functions that are not
submodular.


\subsection{Sum of Metric Distances of a Set}


Let $U$ be a metric space with a distance function $d(\cdot,\cdot)$.
For any subset $S$, define $d(S)$ to be the sum of distances induced by $S$;
i.e., $$d(S)=\sum_{\{u,v\} \subseteq S} d(u,v)$$ where $d(u,v)$ measures the distance between $u$ and $v$. The problem of maximizing $d(S)$ 
(subject to say a cardinality or matroid constraint) is
one of many dispersion problems studied in location theory. 

We also extend the function to a pair of disjoint subsets $S$ and $T$ and define $d(S,T)$ to be the sum of distances between
$S$ and $T$; i.e., \[d(S, T)=\sum_{ u\in S, v\in T } d(u,v).\]

We have the following proposition. 
\begin{proposition}
\label{prop:subclass2}
The sum of metric distances $d(S)$ of a set is weakly submodular (and clearly
monotone).
\end{proposition}
\begin{proof}
Given two subsets $S$ and $T$ of $U$, let $A=S\setminus T$, $B=T\setminus S$ and $C=S\cap T$.
Observe the fact that by the triangle inequality, we have
$$|B|d(A,C)+|A|d(B,C)\ge|C|d(A,B).$$
Therefore,
\begin{eqnarray*}
&  &|T|d(S)+|S|d(T)\\
&=&(|B|+|C|)[d(A)+d(C)+d(A,C)]+(|A|+|C|)[d(B)+d(C)+d(B,C)]\\ 
&=&|C|[d(A)+d(B)+d(C)+d(A,C)+d(B,C)]+(|A|+|B|+|C|)d(C)\\
&  &+|B|d(A)+|A|d(B)+|B|d(A,C)+|A|d(B,C)\\
&\ge&|C|[d(A)+d(B)+d(C)+d(A,C)+d(B,C)]+|S\cup T|d(S\cap T)+|C|d(A,B)\\
&=&|C|[d(A)+d(B)+d(C)+d(A,C)+d(B,C)+d(A,B)]+|S\cup T|d(S\cap T)\\
&=&|S\cap T|d(S\cup T)+|S\cup T|d(S\cap T).
\end{eqnarray*}
\qed
\end{proof}

\subsection{Minimum Cardinality Functions}

For any $k \geq 1$, let $f_k(S) = B > 0$ iff $|S| \geq k$. 

\begin{proposition} 
\label{prop:weak-not-super} 

\begin{enumerate}
\item For $k = 1,2$, $f_k$ is weakly submodular
\item For $k \geq 3, f_k$ is not weakly submodular on any universe of
size at least $k$
\end{enumerate}. 
\end{proposition}

\begin{proof}

In all cases, we need only restrict attention to non empty sets 
$S$ and $T$ in the weak submodularity definition since we are assuming 
$f(\emptyset) = 0$. 
\begin{enumerate}
\item For $k = 1$, weak submodularity follows from the fact that
$|S| + |T| = |S \cap T| + |S \cup T|$ given that $f_1(Z) = B$ for 
all non empty sets $Z$.
\item For $k = 2$, we can verify that $f$ is weakly submodular 
by considering the
possible cardinalities of the sets in the
weakly submodular definition; that is, when
say $|S| \leq |T|$ we consider the cases $|S| < 2$ and $|S| \geq 2$. For
$|S| < 2$, either $S \subseteq T$ or $|S \cap T| = \emptyset$ and we can
easily verify that $f$ satisfies the weak submodularity definition in either
case. If $|S|$ and $|T|$ are both $\geq 2$, then weak submodularity 
follows as in the proof for $k = 1$ since $f_2(Z) = B$ for 
all sets $Z$ with cardinality at least 2.

\item If $k \geq 3$, let $S = \{a_1, \ldots a_{k-2},u\}$ and  
$T = \{b_1, \ldots, b_{k-2},u\}$ for distinct elements 
$a_1 \ldots a_{k-2}, ...b_1 \ldots b_{k-2},u$. Then   
\begin{itemize}
\item $|T| f_k(S) + |S| f_k(T) = 0$
\item
                   $|S \cap T| f_k(S \cup T) + |S \cup T| f_k(S \cap T)$
                    = B
\end{itemize}

\end{enumerate} 
\qed
\end{proof}

\subsection{Average Non-Negative Segmentation Functions}
Motivated by appliations in clustering and data mining, Kleinberg, 
Papadimitriou and Raghavan \cite{KPR09} introduce the general class of
segmentation functions. In their generality, segmentation functions 
need not be submodular nor monotone. They show that every segmentation belongs
to call they call {\it meta-submodular functions} and consider the greedy 
algorithm for ``weakly montone'' 
meta-submodular functions. We now consider another broad class of segmentation 
functions. 

Given an $m\times n$ matrix $M$ and any subset $S\subseteq [m]$, 
a {\em segmentation function} $\sigma(S)$ is the sum of the maximum elements of each column whose row indices appear in $S$; i.e.; $\sigma(S)=\sum_{j=1}^n\max_{i\in S} M_{ij}$.
A segmentation function is {\em average non-negative} if for each row $i$, the sum of all entries of $M$ is non-negative; i.e.,
$\sum_{j=1}^n M_{ij}\ge 0$. 

We can use columns to model individuals, and rows to model items, then each entry of $M_{ij}$ represents how much the individual $j$ likes the item $i$. The average non-negative property basically requires that for each item $i$, on average people do not hate it. Next, we show that an average non-negative segmentation function is weakly-submodular. We first prove the following two lemmas.

\begin{lemma}
An average non-negative segmentation function is monotone.
\end{lemma}
\label{lem:sf1}
\begin{proof}
Let $S$ be a proper subset of $[m]$, and $e$ be an element in $[m]$ that is not in $S$.
If $S$ is empty, then by the average non-negative property, we have $\sigma(\{e\})=\sum_{j=1}^n M_{ej}\ge 0$.
Otherwise, by adding $e$ to $S$ we have $\max_{i\in S\cup\{e\}} M_{ij}\ge\max_{i\in S} M_{ij}$ for all $1\le j\le n$.
Therefore $\sigma(S\cup\{e\})\ge\sigma(S)$.
\qed
\end{proof}

\begin{lemma}
\label{lem:sf2}
For any non-disjoint set $S$ and $T$ and an average non-negative segmentation function $\sigma(\cdot)$,
we have $$\sigma(S)+\sigma(T)\ge\sigma(S\cup T)+\sigma(S\cap T).$$ 
This is also referred as the meta-submodular property~\cite{KPR98}.
\end{lemma}
\begin{proof}
For any non-disjoint set $S$ and $T$ and an average non-negative segmentation function $\sigma(\cdot)$, we let $\sigma_j(S)=\max_{i\in S} M_{ij}$. We show a stronger statement that for any $j\in [n]$, we have
$$\sigma_j(S)+\sigma_j(T)\ge\sigma_j(S\cup T)+\sigma_j(S\cap T).$$
Let $e$ be an element in $S\cup T$ such that $M_{ej}$ is maximum. 
Without loss of generality, assume $e\in S$, then $\sigma_j(S)=\sigma_j(S\cup T)=M_{ej}$. Since $S\cap T\subseteq T$, we have $\sigma_j(T)\ge\sigma_j(S\cap T)$. Therefore,
$$\sigma_j(S)+\sigma_j(T)\ge\sigma_j(S\cup T)+\sigma_j(S\cap T).$$
Summing over all $j\in [n]$, we have
$$\sigma(S)+\sigma(T)\ge\sigma(S\cup T)+\sigma(S\cap T)$$
as desired. 
\qed
\end{proof}

\begin{proposition}
\label{prop:subclass3}
Any average non-negative segmentation function is weakly submodular.
\end{proposition}
\begin{proof}
For any two set $S$ and $T$ and an average non-negative segmentation function $\sigma(\cdot)$, if $S$ and $T$ are non-disjoint then by Lemma~\ref{lem:sf2},
$S$ and $T$ satisfy the submodular property and hence they satisfy the weakly submodular property by Proposition~\ref{prop:subclass1}.
If $S$ and $T$ are disjoint, then $|S\cap T|=0$, and $|S\cup T|=|S|+|T|$.
By monotonicity property in Lemma~\ref{lem:sf1}, we also have
$\sigma(S)\ge\sigma(S\cap T)$ and $\sigma(T)\ge\sigma(S\cap T)$.
Therefore,
$$|S\cap T|\sigma(S\cup T)+|S\cup T|\sigma(S\cap T)\le |T|\sigma(S\cap T)+|S|\sigma(S\cap T)\le |T|\sigma(S)+|S|\sigma(T);$$
the weakly submodular property is also satisfied.
\qed
\end{proof}

\subsection{Small Powers of the Cardinality of a Set}

Clearly, for any positive integer $k$, the functions $f(S) = |S|^k$ can be computed in
time $O(\log k)$. However, given Lemma \ref{lem:nlc} below, it
is still useful to know what simple functions can be used in conjuction with
other submodular and weakly submodular functions. 

It is immediate  to see that the functions 
$f(S)=|S|^0$ and $f(S)=|S|^1$ are linear and hence submodular. We will
show that the square and the cube of the cardinality of a set are
also weakly submodular. 


\begin{proposition}
\label{prop:subclass4}
The square of cardinality of a set is weakly submodular.
\end{proposition}
\begin{proof}
Given two subsets $S$ and $T$ of $U$, let $a=|S\setminus T|$, $b=|T\setminus S|$ and $c=|S\cap T|$.
\begin{eqnarray*}
&  &|T|f(S)+|S|f(T)\\
&=&(b+c)(a+c)^2+(a+c)(b+c)^2\\ 
&=&(a+b+2c)(b+c)(a+c)\\
&=&(a+b+2c)(ab+ac+bc+c^2)\\
&\ge&(a+b+2c)(ac+bc+c^2)\\
&=&(a+b+2c)c(a+b+c)\\
&=&c(a+b+c)^2+(a+b+c)c^2\\
&=&|S\cap T|f(S\cup T)+|S\cup T|f(S\cap T).
\end{eqnarray*}
\qed
\end{proof}


\begin{proposition}
\label{prop:subclass5}
The cube of cardinality of a set is weakly submodular.
\end{proposition}
\begin{proof}
Given two subsets $S$ and $T$ of $U$, let $a=|S\setminus T|$, $b=|T\setminus S|$ and $c=|S\cap T|$.
\begin{eqnarray*}
&  &|T|f(S)+|S|f(T)\\
&=&(b+c)(a+c)^3+(a+c)(b+c)^3\\ 
&=&(a^2+b^2+2c^2+2ac+2bc)(b+c)(a+c)\\
&=&[(a+b+c)^2+c^2-2ab][ab+c(a+b+c)]\\
&=&[(a+b+c)^2+c^2][c(a+b+c)]+ab[(a+b+c)^2+c^2]-2a^2b^2-2abc(a+b+c)\\
&=&c(a+b+c)^3+c^3(a+b+c)+ab[(a+b+c)^2+c^2-2ab-2c(a+b+c)]\\
&=&|S\cap T|f(S\cup T)+|S\cup T|f(S\cap T)+ab(a^2+b^2+c^2+2ab+2ac+2bc+c^2-2ab-2ac-2bc-2c^2)\\
&=&|S\cap T|f(S\cup T)+|S\cup T|f(S\cap T)+ab(a^2+b^2)\\
&\ge&|S\cap T|f(S\cup T)+|S\cup T|f(S\cap T).
\end{eqnarray*}
\qed
\end{proof}
It is easy to see that the function is weakly submodular for $f(S)=|S|^0$ and $f(S)=|S|^1$. We now give an example that shows $f(S)=|S|^4$ is not
weakly submodular. 

\subsubsection{Higher powers}
\begin{proposition} $f(S)=|S|^4$ is not
weakly submodular.
\end{proposition}

\begin{proof} Given two subsets $S$ and $T$ of $U$, let $a=|S\setminus T|$, $b=|T\setminus S|$ and $c=|S\cap T|$. Suppose $a=4, b=4, c=1$.

$$|T|f(S)+|S|f(T)=(b+c)(a+c)^4+(a+c)(b+c)^4=6250$$
On the other hand, we have
$$|S\cap T|f(S\cup T)+|S\cup T|f(S\cap T)=c(a+b+c)^4+c^4(a+b+c)=9^4+9=6570$$
Therefore, the function is not weakly submodular. 
\end{proof}

Similarly, one can see that  $f(S| = |S|^k$ is not weakly submodular
for all intergers $k \geq 4$. 

\subsection{Linear combinations of weakly submodular functions}
Next we show a basic but important property of weakly submodular functions.
\begin{lemma}
\label{lem:nlc}
Non-negative linear combinations of weakly submodular functions are weakly submodular.
\end{lemma}
\begin{proof}
Consider weakly submodular functions $f_1,f_2,\dots,f_n$ and non-negative numbers $\alpha_1, \alpha_2,\dots,\alpha_n$. 
Let $g(S)=\sum_{i=1}^n \alpha_i f_i(S)$, then for any two set $S$ and $T$, we have
\begin{eqnarray*}
&  &|T|g(S)+|S|g(T)\\
&=&|T|\sum_{i=1}^n\alpha_if_i(S)+|S|\sum_{i=1}^n\alpha_if_i(T)\\ 
&=&\sum_{i=1}^n\alpha_i[|T|f_i(S)+|S|f_i(T)]\\
&\ge&\sum_{i=1}^n\alpha_i[|S\cap T|f_i(S\cup T)+|S\cup T|f_i(S\cap T)]\\
&=&|S\cap T|\sum_{i=1}^n\alpha_if_i(S\cup T)+|S\cup T|\sum_{i=1}^n\alpha_if_i(S\cap T)\\
&=&|S\cap T|g(S\cup T)+|S\cup T|g(S\cap T).
\end{eqnarray*}
Therefore, $g(S)$ is weakly submodular.
\qed
\end{proof}

\begin{corollary}
\label{cor:ca-matroid}
The welfare maximization
problem (also known as the maximization problem for combinatorial auctions) 
for agents with weakly submodular valuations is a special case
of the maximization of a weakly submodular function subject to a 
partition matroid.

\end{corollary}

\begin{proof}

In the maximum welfare problem, $n$ agents $A = \{1, \ldots, n\}$ 
have valuation functions 
$v_i :=  U \rightarrow \Re$. A feasible allocation is a disjoint allocation 
of subsets $S_i$ to each agent ($1 \leq i \leq n$) so as to maximize  
the social welfare function $f(S) = \sum_{i = 1}^{n} v_i(S_i)$. 
It is well known then how to view the disjointness constraint as  
a partition matroid constraint. Namely, we consider a universe
$U' = A \times U$ where the elements 
of $U'$ are partitioned into blocks $B_u = \{i,u)| i \in A\}$ 
for each $u \in U$. For $S' = \cup B'_u$, we let the partition matroid 
be defined by the 
independence  
condition that a subset $S' \subseteq U'$ is independent iff $|B'_u| \leq 1$;
that is, it does 
not contain any two elements $(i,u)$ and $(i',u)$ for some $u \in U$ and 
$i \neq i'$. Letting $\pi_i(S') = \{u: (i,u) \in S'\}$, 
 define $f'_i(S) = v_i(\pi_i(S')$ and 
$f'(S') = \sum_{i = 1}^{n} v_i(\pi_i(S'))$ for any subset $S' \subseteq U'$. 
Given that each $v_i$ is weakly submodular on the universe $\pi_i(U')$ and
that  the class 
of weakly submodular
functions is closed under linear combinations, $f'(S')$ is a weakly submodular
function when all the valuations $v_i$ are weakly submodular. 

\qed 
\end{proof}

We now show two more examples of weakly submodular function using Lemma~\ref{lem:nlc}.

\subsection{The Objective Function of Max-Sum Diversification}
\begin{corollary}
\label{cor:msd}
The objective function of the max-sum diversification problem, 
$f(S) = g(S) + \sum_{\{u,v\} \subseteq S} d(u,v)$, is weakly submodular 
when $g$ is submodular (or weakly submodular) and $d$ is a metric.
\end{corollary}
\begin{proof}
This follows immediate from Proposition~\ref{prop:subclass1} and \ref{prop:subclass2} and Lemma~\ref{lem:nlc}. 
\qed
\end{proof}

\subsection{Restricted Polynomial Function on the Cardinality of a Set}
\begin{corollary}
\label{cor:msd}
For polynomial function on the cardinality of a set, if the degree is less than four and coefficients are all non-negative, then the function is weakly submodular.
\end{corollary}
\begin{proof}
This follows immediate from Proposition~\ref{prop:subclass4} and \ref{prop:subclass5} and Lemma~\ref{lem:nlc}. 
\qed
\end{proof}

\section{Related Work}
\label{sec:related-work}

Recently, there has been other     
generalizations of monotone submodular functions\footnote{We note that the class of weakly submodular functions was introduced in 
the PHD thesis of Ye \cite{Ye13} and followed from observations made with 
 regard to the diversifcation problem in \cite{BLY12}. As such this class 
was studied independently from the work 
relating to supermodular degree and the MPH-$k$ hierarchy that will 
now be discussed.}. In particular with regard to 
combinatorial auctions, Feige and Izsak \cite{FeigeI13} 
defined the concept of the {\it supermodularity degree} of a set function as a 
measure of the degree of complementarity. Intuitvely, 
for each item $u$, its supermodular
degree is the number of other items $v$ that increase the marginal 
value of  
$u$ with respect to some subset not containing $u$.
This induces a supermodular dependency graph and the supermodular degree
of an item is its degree in this dependency graph. 
The supermodular degree of a set function is the maximum  
of the item supermodular degrees. Set functions with supermodular degree 0 
are precisely the  
submodular functions and every set function on a universe $U$ 
 has supermodular degree at most $|U|-1$. 
Feige and Izsak consider the welfare maximization
 problem 
when each agent 
has a valuation function with supermodular degree at most $d$. 
Amongst their results, they show that given the 
supermodular dependency graph, and a value oracle to access the valuation
function of each agent, a greedy algorithm 
approximates the welfare maximization problem 
to within a factor \footnote{We are stating all of our approximation ratios 
to be greater than or equal to 1 whereas Feige and Izsak use fractional
approximation ratios} of $d+2$. Feldman and Izsak \cite{FeldmanI14}
consider the maximization of set functions with supermodular degree $d$ 
degree subject to independence in a matroid and more generally to
independence in a $k$-extendible system as defined by Mestre \cite{Mestre06}.
They show that a natural greedy algorithm achieves approximation ratio 
$k(d+1)+1$ assuming
a value oracle (for accessing the set function) and an independence 
oracle (for determining if a set is indepedendent in ${\cal I}$).

It is easy to see that the class of weakly submodular functions does
not correspond to functions having bounded supermodular degree. 
On the one hand, the function $f_2$ in Proposition \ref{prop:weak-not-super}
is weakly submodular and has supermodular degree $|U| - 1$ for
any universe $U$ with at least 3 elements. 
Furthermore, Feige et al \cite{FeigeFIILS14} show that 
there are instances of the metric sum dispersion problem 
(even with unit distance 
on the complete graph $G = (U, U \times U)$ ) that does not have bounded 
supermodular degree. In fact, Feige et al show that for this instance
of the dispersion function, a function of supermodular degree $d$ 
cannot provide an approximation better than $\frac{|U|}{d+1} -1$. 
On the other  hand, we have the following observation: 

\begin{proposition}



There are simple functions having supermodular degree 1
that are not weakly submodular. Namely, for the universe $U = \{a_1,a_2,b\}$,
let $f(S) = B > 0$ if $\{a_1,a_2\} \subseteq S$ and 0 otherwise. 
Letting  $S = \{a_1,b_1\}$ and  $T = \{a_2,b_1\}$, we have 
             \begin{itemize}
\item $|T| f(S) + |S| f(T) = 0$ 
\item
                   $|S \cap T| f(S \cup T) + |S \cup T| f(S \cap T)$
                    = B
\end{itemize}
which violates the definition of weak submodularity.
\qed
\end{proposition}

Another generalization of submodular functions was introduced in 
Conitzer, Sandholm and Santi \cite{ConitzerSS05} and further developed
in the expessive MPH-$k$ hierarchy of Feige et al \cite{FeigeFIILS14}. 
They consider the representation of a set function $f(S)$ by its 
unique hypergraph $h(S)$  (called hypercube in \cite{ConitzerSS05}) 
representation. Functions in which the only non zero elements $h(S)$
in the hypergraph representation are positive and further satisfy $|S| \leq k$ 
are called PH-$k$ functions. A monotone function is in the 
class MPH-$k$ if it can be expressed as maximum over a finite  collection 
of PH-$k$ functions. Feige et al establish a number of significant results
amongst which (most relevant to our results)  are the facts that all 
monotone functions of supermodular degree 
$k-1$ are in MPH-$k$ for $k \geq 1$ and that using demand oracles 
and given the hypergraph 
representation of agent set functions, the welfare maximization problem 
for agents  with MPH-$k$ valuations can be solved by an LP-rounding 
algorithm with approximation ratio $k+1$. As a special case, we note that 
the sum dispersion problem is a MPH-$2$ function (even for non metric
distances). As they show (in their appendix L), 
the expressiveness of the MPH-$k$ framework  
may require some simple functions (even in MPH-$1$) to require exponentially  
many hypergraphs to be so represented. While functions in any MPH-$k$ are 
closed under linear combinations, maximizing such functions to a cardinality
constraint (and hence to matroid constraints) would require a breakthrough
for the densest subgraph problem  since the densest subgraph problem 
subject to a cardinality constraint can be reduced to the MPH-$2$ non metric dispersion problem (see 
Feige, Kortsarz and Peleg \cite{FeigePK01}, Andersen and Chellapilla 
\cite{AndersenC09} and  
Khuller and Saha \cite{KhullerS09}). 

\section{Weakly Submodular Function Maximization Subject to a Cardinality 
Constraint}
\label{sec:cardinality}

We emphasize again that we restrict attention to 
monotone, non-negative and normalized functions. 
In this section, we discuss a greedy approximation algorithm for maximizing 
weakly submodular functions subject to a uniform  matroid (i.e cardinality 
constraint). In section \ref{ssec:wsub-fd} we consider an arbitrary 
matroid constraint. 

Given an underlying set $U$ and a weakly submodular function $f(\cdot)$ defined on
every subset of $U$, the goal is to select a subset $S$ maximizing $f(S)$ subject to a cardinality constraint $|S|\le p$.
We consider the following standard greedy algorithm that achieves 
approximation ratio $\frac{e}{e-1}$ for monotone submodular
maximization by a classic result of Nemhauser, Fisher and 
Wolsey \cite{NWF78}. Furthermore, they showed that this is the best approximationpossible in the value oracle model and Feige \cite{Feige98} showe the 
same inapproximation holds for an explictly defined function subject to
the conjecture that $RP \neq NP$. By a result of Birnbaum and Goldman 
\cite{BirnbaumG09}, it is known that the same greedy algorithm is a 
2-approximation for the metric dispersion problem subject to a 
cardinality constraint  
as well as a 2-approximation
for the cardinality constrained diversification problem in Borodin et al \cite{BJLY14, BLY12}. 
\medskip

\noindent{\sc Greedy Algorithm for Weakly Submodular Function Maximization}

\begin{algorithmic}[1] 
\STATE $S=\emptyset$
\WHILE{$|S|<p$}
\STATE Find $u\in U\setminus S$ maximizing $f(S\cup\{u\})-f(S)$  
\STATE $S=S\cup\{u\}$
\ENDWHILE
\STATE return $S$
\end{algorithmic}

\begin{theorem}
\label{thm:wsgreedy}
The standard greedy algorithm achieves approximation ratio $\approx 5.95$.
\end{theorem}

Before getting into the proof, we first prove two algebraic identities.

\begin{lemma}
\label{lem:id1}
$$\sum_{j=1}^n(\frac{i+1}{i})^{j-1}=i(\frac{i+1}{i})^n-i.$$
\end{lemma}
\begin{proof}
Note that the expression on the left-hand side is a geometric sum. Therefore, we have
$$\sum_{j=1}^n(\frac{i+1}{i})^{j-1}=\frac{(\frac{i+1}{i})^n-1}{\frac{i+1}{i}-1}=i(\frac{i+1}{i})^n-i.$$
\qed
\end{proof}

\begin{lemma}
\label{lem:id2}
$$\sum_{j=1}^nj(\frac{i+1}{i})^{j-1}=ni^2(\frac{i+1}{i})^{n+1}-(n+1)i^2(\frac{i+1}{i})^n+i^2.$$
\end{lemma}
\begin{proof}
Consider the function $f(x)=\sum_{j=1}^nx^j$ with $x\ne 1$, its derivative $f'(x)=\sum_{j=1}^njx^{j-1}$.
Since $f(x)$ is a geometric sum and $x\ne 1$, we have 
$$f(x)=\frac{x^{n+1}-1}{x-1}.$$ 
Taking derivatives on both sides we have
$$f'(x)=\frac{(n+1)x^n(x-1)-x^{n+1}+1}{(x-1)^2}=\frac{nx^{n+1}-(n+1)x^n+1}{(x-1)^2}.$$ 
Therefore, we have
$$\sum_{j=1}^njx^{j-1}=\frac{nx^{n+1}-(n+1)x^n+1}{(x-1)^2}.$$ 
Substituting $x$ with $\frac{i+1}{i}$, we have
$$\sum_{j=1}^nj(\frac{i+1}{i})^{j-1}=\frac{n(\frac{i+1}{i})^{n+1}-(n+1)(\frac{i+1}{i})^n+1}{(\frac{i+1}{i}-1)^2}=ni^2(\frac{i+1}{i})^{n+1}-(n+1)i^2(\frac{i+1}{i})^n+i^2.$$ 
\qed
\end{proof}

Now we proceed to the proof to Theorem~\ref{thm:wsgreedy}.
\begin{proof}
Let $S_i$ be the greedy solution after the $i^{\rm th}$ iteration; i.e., $|S_i|=i$. Let $O$ be an optimal solution,
and let $C_i=O\setminus S_i$. Let $m_i=|C_i|$, and $C_i=\{c_1, c_2, \dots, c_{m_i}\}$. By the weakly submodularity definition, we get the following $m_i$ inequalities for each $0<i<p$:
\begin{align*}
(i+m_i-1)f(S_i\cup\{c_1\})+(i+1)f(S_i\cup\{c_2,\dots,c_{m_i}\}) &\ge (i)f(S_i\cup\{c_1\dots,c_{m_i}\})
  +(i+m_i)f(S_i)\\
(i+m_i-2)f(S_i\cup\{c_2\})+(i+1)f(S_i\cup\{c_3,\dots,c_{m_i}\}) &\ge (i)f(S_i\cup\{c_2\dots,c_{m_i}\})
  +(i+m_i-1)f(S_i)\\
  &\vdots\\
(i+1)f(S_i\cup\{c_{m_i-1}\})+(i+1)f(S_i\cup\{c_{m_i}\})&\ge (i)f(S_i\cup\{c_{m_i-1},c_{m_i}\})
+(i+2)f(S_i)\\
(i)f(S_i\cup\{c_{m_i}\})+(i+1)f(S_i)&\ge (i)f(S_i\cup\{c_{m_i}\})+(i+1)f(S_i).
\end{align*}

Multiplying the $j^{\rm th}$ inequality by $(\frac{i+1}{i})^{j-1}$, and summing all of them up (noting that the second term of the left hand 
side of the $j^{th}$ inequality then cancels the first term of the $j+1^{st}$ 
inequality), we have
\begin{align*}
\sum_{j=1}^{m_i}(i+m_i-j)(\frac{i+1}{i})^{j-1}f(S_i\cup\{c_j\})&+(i+1)(\frac{i+1}{i})^{m_i-1}f(S_i)\\
&\ge(i)f(S_i\cup\{c_1,\dots,c_{m_i}\})+\sum_{j=1}^{m_i}(i+m_i-j+1)(\frac{i+1}{i})^{j-1}f(S_i).
\end{align*}

By monotonicity, we have $f(S_i\cup\{c_1,\dots,c_{m_i}\})\ge f(O)$. Rearranging the inequality,
$$\sum_{j=1}^{m_i}(i+m_i-j)(\frac{i+1}{i})^{j-1}f(S_i\cup\{c_j\})\ge(i)f(O)+\sum_{j=1}^{m_i-1}(i+m_i-j+1)(\frac{i+1}{i})^{j-1}f(S_i).$$
By the greedy selection rule, we know that 
$f(S_{i+1})\ge f(S_i\cup\{c_j\})$ for any $1\le j\le m_i$, therefore we have
$$\sum_{j=1}^{m_i}(i+m_i-j)(\frac{i+1}{i})^{j-1}f(S_{i+1})\ge(i)f(O)+\sum_{j=1}^{m_i-1}(i+m_i-j+1)(\frac{i+1}{i})^{j-1}f(S_i).$$
For the ease of notation, we let
\begin{align*}
a_i=\sum_{j=1}^{m_i}(i+m_i-j)(\frac{i+1}{i})^{j-1}&& b_i=\sum_{j=1}^{m_i-1}(i+m_i-j+1)(\frac{i+1}{i})^{j-1}
\end{align*}
so that we have $a_i f(S_{i+1}) - b_i f(S_i) \geq (i) f(O)$ \\
 
We first simplify $a_i$ and $b_i$. 
\begin{align*}
a_i&=\sum_{j=1}^{m_i}(i+m_i-j)(\frac{i+1}{i})^{j-1}\\
&=\sum_{j=1}^{m_i}(i+m_i)(\frac{i+1}{i})^{j-1}-\sum_{j=1}^{m_i}j(\frac{i+1}{i})^{j-1}.
\end{align*}
By Lemma~\ref{lem:id1} and \ref{lem:id2}, we have
\begin{align*}
a_i&=(i+m_i)[i(\frac{i+1}{i})^{m_i}-i]-m_ii^2(\frac{i+1}{i})^{m_i+1}+(m_i+1)i^2(\frac{i+1}{i})^{m_i}-i^2\\
&=[i^2+im_i-m_i(i^2+i)+(m_i+1)i^2](\frac{i+1}{i})^{m_i}-2i^2-im_i\\
&=2i^2(\frac{i+1}{i})^{m_i}-2i^2-im_i.
\end{align*}
Similarly, we have
\begin{align*}
b_i&=\sum_{j=1}^{m_i-1}(i+m_i-j+1)(\frac{i+1}{i})^{j-1}\\
&=\sum_{j=1}^{m_i-1}(i+m_i+1)(\frac{i+1}{i})^{j-1}-\sum_{j=1}^{m_i-1}j(\frac{i+1}{i})^{j-1}\\
&=(i+m_i+1)[i(\frac{i+1}{i})^{m_i-1}-i]-(m_i-1)i^2(\frac{i+1}{i})^{m_i}+m_ii^2(\frac{i+1}{i})^{m_i-1}-i^2\\
&=[i^2+im_i+i-(m_i-1)(i^2+i)+m_ii^2](\frac{i+1}{i})^{m_i-1}-2i^2-im_i-i\\
&=2i(i+1)(\frac{i+1}{i})^{m_i-1}-2i^2-im_i-i\\
&=2i^2(\frac{i+1}{i})^{m_i}-2i^2-im_i-i.
\end{align*}
Now let 
\begin{align*}
a^*_i=\sum_{j=1}^{p}(i+p-j)(\frac{i+1}{i})^{j-1}&& b^*_i=\sum_{j=1}^{p-1}(i+p-j+1)(\frac{i+1}{i})^{j-1}
\end{align*}
The simplication of $a_i$ and $b_i$ makes it clear that 
$a_i-b_i = i$ for any value of $m_i$. Since $a^*_i$ (resp. $b^*_i$) can
be thought of as $a_i$ (resp. $b_i$) with $m_i = p$, we have
$$a^*_i-a_i=b^*_i-b_i\ge 0$$
Therefore,
$$a^*_if(S_{i+1})-b^*_if(S_i)=a_if(S_{i+1})-b_if(S_i)+(a^*_i-a_i)[f(S_{i+1})-f(S_i)].$$
Since $f(\cdot)$ is monotone, we have $f(S_{i+1})-f(S_i)\ge 0$. Therefore,
$$a^*_if(S_{i+1})-b^*_if(S_i)\ge a_if(S_{i+1})-b_if(S_i)\ge if(O).$$
Then we have the following set of inequalities:
\begin{align*}
a^*_1f(S_2)&\ge 1f(O)+b^*_1f(S_1)\\
a^*_2f(S_3)&\ge 2f(O)+b^*_2f(S_2)\\
&\vdots\\
a^*_{p-2}f(S_{p-1})&\ge (p-2)f(O)+b^*_{p-2}f(S_{p-2})\\
a^*_{p-1}f(S_p)&\ge (p-1)f(O)+b^*_{p-1}f(S_{p-1}).
\end{align*}

Multiplying the $i^{\rm th}$ inequality by $\frac{\prod_{j=1}^{i-1}a^*_j}{\prod_{j=2}^{i}b^*_j}$, summing all of them up and ignoring the term $b^*_1f(S_1)$,
$$\frac{\prod_{j=1}^{p-1}a^*_j}{\prod_{j=2}^{p-1}b^*_j}f(S_p)\ge\sum_{i=1}^{p-1}\frac{i\prod_{j=1}^{i-1}a^*_j}{\prod_{j=2}^{i}b^*_j}f(O).$$
Therefore the approximation ratio
$$\frac{ f(O)}{ f(S_p)}\le \frac{\frac{\prod_{j=1}^{p-1}a^*_j}{\prod_{j=2}^{p-1}b^*_j}}{\sum_{i=1}^{p-1}\frac{i\prod_{j=1}^{i-1}a^*_j}{\prod_{j=2}^{i}b^*_j}}
=\left(\sum_{i=1}^{p-1}\frac{i\prod_{j=i+1}^{p-1}b^*_j}{\prod_{j=i}^{p-1}a^*_j}\right)^{-1}
=\left(\sum_{i=1}^{p-1}\left[\frac{i}{a^*_i}\cdot\prod_{j=i+1}^{p-1}\frac{b^*_j}{a^*_j}\right]\right)^{-1}.$$
Note that the approximation ratio is simply a function of $p$.
In particular, the approximation ratio is 3.74 when $p=10$ and approximation ratio is 5.62 when $p=100$. Computer evaluations suggest that the approximation 
ratio 
converges to 5.95 as $p$ tends to $\infty$.  
\qed
\end{proof}

In terms of hardness of approximation, assuming $P \neq NP$, 
Feige \cite{Feige98} proved that 
the max coverage problem 
(an example of monotone submodular maximization subject to a 
cardinality constraint) is
known to be hard to approximate to a factor better than 
$\frac{e}{e-1} - \epsilon$.  The problem of maximizing the sum of metric 
distances subject to a cardinality constraint has been called the
{\it max-sum dispersion problem}.  
The max-sum dispersion problem is known to be NP-hard
by an easy reduction from Max-Clique, and as noted by 
Alon \cite{Alon14}, there is evidence that the problem is hard to 
compute in polynomial time with approximation
$2-\epsilon$ for any $\epsilon > 0$
when $p = n^r$ for $1/3 \leq r < 1$.  (See the discussion in Section 3 of 
\cite{BJLY14}.)
%


\section{Weakly Submodular Function Maximization Subject to an Arbitrary Matroid
Constraint}
\label{ssec:wsub-fd}
It is natural to consider a general matroid constraint for the problem of weakly submodular function maximization. For this more general problem, the greedy algorithm in the previous section no longer achieves any constant approximation ratio. 
See the example presented in the Appendix of \cite{BJLY14}.
 Following the result for
max-sum diversification subject to a matroid constraint 
in \cite{BLY12}, we will analyze the following oblivious local search algorithm:
\medskip

\noindent{\sc Weakly Submodular Function Maximization with a Matroid Constraint}

\begin{algorithmic}[1] 
\STATE Let $S$ be a basis of $\cal M$
\WHILE{exists $u\in U\setminus S$ and $v\in S$ such that $S\cup\{u\}\setminus\{v\}\in {\cal F}$ and $f(S\cup\{u\}\setminus\{v\}) >f(S)$}
\STATE $S=S\cup\{u\}\setminus\{v\}$
\ENDWHILE
\STATE return $S$
\end{algorithmic}

The following lemma on the exchange property of matroid bases was first stated in \cite{Bru69}.
\begin{lemma}[Brualdi \cite{Bru69}]
For any two sets $X,Y \in {\cal F}$ with $|X| =|Y|$, there is a bijective mapping
$g : X \rightarrow Y$ such that $X \cup \{g(x)\} \setminus \{x\} \in {\cal F}$ for any $x \in X$.
\label{lem:bj}
\end{lemma}

Before we prove the theorem, we need to prove several lemmas.
Let $O$ be the optimal solution, and $S$, the solution at the end of the local search algorithm.
Let $s$ be the size of a basis; let
$A=O\cap S$, $B=S\setminus A$ and $C=O\setminus A$.
By Lemma~\ref{lem:bj}, there is a bijective mapping $g: B\rightarrow C$ such that $S\cup\{b\}\setminus\{g(b)\}\in {\cal F}$ for any $b\in B$.
Let $B=\{b_1, b_2, \dots, b_t\}$, and let $c_i=g(b_i)$ for all $i=1,\dots, t$.
We reorder $b_1, b_2, \dots, b_t$ in different ways.
Let $b'_1, b'_2, \dots, b'_t$ be an ordering such that the corresponding $c'_1, c'_2, \dots, c'_t$ 
maximizes the sum $\sum_{i=1}^t(s-i)(\frac{s+1}{s})^{i-1} f(S\cup\{c'_i\})$;
and let $b''_1, b''_2, \dots, b''_t$ be an ordering such that the corresponding $c''_1, c''_2, \dots, c''_t$ 
minimizes the sum \[\sum_{i=1}^t(s+t-i)(\frac{s+1}{s})^{i-1} f(S\cup\{c''_i\}).\]

\begin{lemma} 
\label{lem:ag}
Given three non-increasing non-negative sequences:
\begin{align*}
&\alpha_1\ge \alpha_2\ge \dots\ge \alpha_n\ge 0,\\
&\beta_1\ge \beta_2\ge \dots\ge \beta_n\ge 0,\\
&x_1\ge x_2\ge \dots\ge x_n\ge 0.
\end{align*}
Then we have
$$\sum_{i=1}^n \alpha_ix_i\sum_{i=1}^n \beta_i\ge\sum_{i=1}^n \beta_ix_{n+1-i}\sum_{i=1}^n \alpha_i.$$
\end{lemma}
\begin{proof}
Consider the following:
\begin{eqnarray*}
n\sum_{i=1}^n \alpha_ix_i&=&n\alpha_1x_1+n\alpha_2x_2+\dots+n\alpha_nx_n\\
&=&\sum_{i=1}^n\alpha_ix_1+(n\alpha_1-\sum_{i=1}^n\alpha_i)x_1+n\alpha_2x_2+\dots+n\alpha_nx_n\\ 
&\ge&\sum_{i=1}^n\alpha_ix_1+(n\alpha_1+n\alpha_2-\sum_{i=1}^n\alpha_i)x_2+\dots+n\alpha_nx_n\\ 
&=&\sum_{i=1}^n\alpha_ix_1+\sum_{i=1}^n\alpha_ix_2+(n\alpha_1+n\alpha_2-2\sum_{i=1}^n\alpha_i)x_2+\dots+n\alpha_nx_n\\ 
&  &\vdots\\
&\ge&\sum_{i=1}^n\alpha_ix_1+\sum_{i=1}^n\alpha_ix_2+\dots+\sum_{i=1}^n\alpha_ix_n+(n\alpha_1+n\alpha_2+\dots+n\alpha_n-n\sum_{i=1}^n\alpha_i)x_n\\ 
&=&\sum_{i=1}^n\alpha_i\sum_{i=1}^nx_i
\end{eqnarray*}
Similarly, we have
\begin{eqnarray*}
n\sum_{i=1}^n \beta_ix_{n+1-i}&=&n\beta_1x_n+n\beta_2x_{n-1}+\dots+n\beta_nx_1\\
&=&\sum_{i=1}^n\beta_ix_n+(n\beta_1-\sum_{i=1}^n\beta_i)x_n+n\beta_2x_{n-1}+\dots+n\beta_nx_1\\ 
&\le&\sum_{i=1}^n\beta_ix_n+(n\beta_1+n\beta_2-\sum_{i=1}^n\beta_i)x_{n-1}+\dots+n\beta_nx_1\\ 
&=&\sum_{i=1}^n\beta_ix_n+\sum_{i=1}^n\beta_ix_{n-1}+(n\beta_1+n\beta_2-2\sum_{i=1}^n\beta_i)x_{n-1}+\dots+n\beta_nx_1\\ 
&  &\vdots\\
&\le&\sum_{i=1}^n\beta_ix_n+\sum_{i=1}^n\beta_ix_{n-1}+\dots+\sum_{i=1}^n\beta_ix_1+(n\alpha_1+n\beta_2+\dots+n\beta_n-n\sum_{i=1}^n\beta_i)x_1\\ 
&=&\sum_{i=1}^n\beta_i\sum_{i=1}^nx_i
\end{eqnarray*}
Therefore the lemma follows.
\qed
\end{proof}

\begin{lemma} 
\label{lem:ie1}
\begin{align*}
\sum_{i=1}^t(s-i) &(\frac{s+1}{s})^{i-1}  f(S\cup\{c'_i\})\\
&\le s f(S)+\sum_{i=1}^t(s+1-i)(\frac{s+1}{s})^{i-1} f(S\cup\{c'_i\}\setminus\{b'_i\})-(s+1)(\frac{s+1}{s})^{t-1} f(S\setminus\{b'_1, \dots, b'_t\}).
\end{align*}
\end{lemma}
\begin{proof}
By the definition of weakly submodular, we have
\begin{align*}
s f(S)+s f(S\cup\{c'_1\}\setminus\{b'_1\}) &\ge (s-1) f(S\cup\{c'_1\})+(s+1) f(S\setminus\{b'_1\})\\
s f(S\setminus\{b'_1\})+(s-1) f(S\cup\{c'_2\}\setminus\{b'_2\}) &\ge (s-2) f(S\cup\{c'_2\})+(s+1) f(S\setminus\{b'_1,b'_2\})\\
&\vdots\\
s f(S\setminus\{b'_1,\dots, b'_{t-1}\})+(s-t+1) f(S\cup\{c'_t\}\setminus\{b'_t\})&\ge (s-t) f(S\cup\{c'_t\})+(s+1) f(S\setminus\{b'_1,\dots, b'_t\})
\end{align*}
Multiplying the $i^{\rm th}$ inequality by $(\frac{s+1}{s})^{i-1}$, and summing all of them up to get
\begin{align*}
s f(S)&+\sum_{i=1}^t(s+1-i)(\frac{s+1}{s})^{i-1} f(S\cup\{c'_i\}\setminus\{b'_i\}) \\
&\ge\sum_{i=1}^t(s-i)(\frac{s+1}{s})^{i-1} f(S\cup\{c'_i\})+(s+1)(\frac{s+1}{s})^{t-1} f(S\setminus\{b'_1, \dots, b'_t\}).
\end{align*}
After rearranging the inequality, we get
\begin{align*}
\sum_{i=1}^t(s-i) &(\frac{s+1}{s})^{i-1} f(S\cup\{c'_i\})\\
&\le s f(S)+\sum_{i=1}^t(s+1-i)(\frac{s+1}{s})^{i-1} f(S\cup\{c'_i\}\setminus\{b'_i\})-(s+1)(\frac{s+1}{s})^{t-1} f(S\setminus\{b'_1, \dots, b'_t\}).
\end{align*}
\qed
\end{proof}

\begin{lemma} 
\label{lem:ie2}
\begin{align*}
\sum_{i=1}^t(s+t-i)(\frac{s+1}{s})^{i-1} f(S\cup\{c''_i\}) &-\sum_{i=1}^t(s+t+1-i)(\frac{s+1}{s})^{i-1} f(S)\\
&\ge s f(S\cup\{c''_1, \dots, c''_t\})-(s+1)(\frac{s+1}{s})^{t-1} f(S)
\end{align*}
\end{lemma}
\begin{proof}
By the definition of weakly submodular, we have
\begin{align*}
(s+t-1) f(S\cup\{c''_1\})+(s+1) f(S\cup\{c''_2,\dots,c''_{m_i}\})&\ge s f(S\cup\{c''_1,\dots,c''_{m_i}\})+(s+t) f(S)\\
&\vdots\\
(s+1) f(S\cup\{c''_{t-1}\})+(s+1) f(S\cup\{c''_t\})&\ge s f(S\cup\{c''_{t-1},c''_t\})+(s+2) f(S)\\
s f(S\cup\{c''_t\})+(s+1) f(S)&\ge s f(S\cup\{c''_t\})+(s+1) f(S).\\
\end{align*}
Multiplying the $i^{\rm th}$ inequality by $(\frac{s+1}{s})^{i-1}$, and summing all of them up, we have
\begin{align*}
\sum_{i=1}^t(s+t-i)(\frac{s+1}{s})^{i-1} f(S\cup\{c''_i\})&+(s+1)(\frac{s+1}{s})^{t-1} f(S)\\
&\ge s f(S\cup\{c''_1, \dots, c''_t\})+\sum_{i=1}^t(s+t+1-i)(\frac{s+1}{s})^{i-1} f(S).
\end{align*}
Therefore, we have
\begin{align*}
\sum_{i=1}^t(s+t-i)&(\frac{s+1}{s})^{i-1} f(S\cup\{c''_i\})\\
&\ge s f(S\cup\{c''_1, \dots, c''_t\})+\sum_{i=1}^t(s+t+1-i)(\frac{s+1}{s})^{i-1} f(S)-(s+1)(\frac{s+1}{s})^{t-1} f(S).
\end{align*}
\qed
\end{proof}

Let 
\begin{align*}
&W=\sum_{i=1}^t(s-i)(\frac{s+1}{s})^{i-1}, && X=\sum_{i=1}^t(s+1-i)(\frac{s+1}{s})^{i-1}, \\
&Y=\sum_{i=1}^t(s+t-i)(\frac{s+1}{s})^{i-1}, && Z=\sum_{i=1}^t(s+t+1-i)(\frac{s+1}{s})^{i-1}.
\end{align*}

\begin{lemma} 
\label{lem:ie3}
$$C\sum_{i=1}^t(s-i)(\frac{s+1}{s})^{i-1} f(S\cup\{c'_i\})\ge A\sum_{i=1}^t(s+t-i)(\frac{s+1}{s})^{i-1} f(S\cup\{c''_i\}).$$
\end{lemma}
\begin{proof}
This is immediate by Lemma~\ref{lem:ag}
\qed
\end{proof}

\begin{theorem}
\label{thm:local-search-approx}
Let $s$ be the size of a basis,
the local search algorithm achieves an approximation ratio bounded by $14.5$ for an arbitrary 
$s$, approximately $10.88$ when $s=6$. The ratio converges to $10.22$ as $s$ tends to $\infty$.
\end{theorem}
\begin{proof}
Since $S$ is a locally optimal solution, we have 
$$ f(S)\ge f(S\cup\{c'_i\}\setminus\{b'_i\}).$$
Since $ f(S\setminus\{b'_1, \dots, b'_t\})\ge 0$, by Lemma~\ref{lem:ie1}, we have
$$\sum_{i=1}^t(s-i)(\frac{s+1}{s})^{i-1} f(S\cup\{c'_i\})\le s f(S)+\sum_{i=1}^t(s+1-i)(\frac{s+1}{s})^{i-1} f(S).$$
Therefore,
$$\sum_{i=1}^t(s-i)(\frac{s+1}{s})^{i-1} f(S\cup\{c'_i\})\le (s+X) f(S).$$
On the other hand, we have $O\subseteq S\cup\{c''_1, \dots, c''_t\}$, by monotonicity, we have
$ f(O)\le  f(S\cup\{c''_1, \dots, c''_t\})$. By Lemma~\ref{lem:ie2}, we have
$$\sum_{i=1}^t(s+t-i)(\frac{s+1}{s})^{i-1} f(S\cup\{c''_i\})\ge s f(O)+[Z-(s+1)(\frac{s+1}{s})^{t-1}] f(S).$$
Lemma~\ref{lem:ag}, we have
$$Y\sum_{i=1}^t(s-i)(\frac{s+1}{s})^{i-1} f(S\cup\{c'_i\})\ge W\sum_{i=1}^t(s+t-i)(\frac{s+1}{s})^{i-1} f(S\cup\{c''_i\}).$$
Therefore
$$Y(s+X) f(S)\ge Ws f(O)+X[Z-(s+1)(\frac{s+1}{s})^{t-1}] f(S)$$
Hence the approximation ratio:
$$\frac{ f(O)}{ f(S)}\le\frac{YX-WZ+Ys+W(s+1)(\frac{s+1}{s})^{t-1}}{Ws}=\frac{YX-WZ+Ys}{Ws}+(\frac{s+1}{s})^t.$$
Simplifying the notation, we have
$$\frac{ f(O)}{ f(S)}\le\frac{\sum_{i=1}^t (s^2+st+ti-si)(\frac{s+1}{s})^{i-1}+\sum_{i=t+1}^{2t-1}t(2t-i)(\frac{s+1}{s})^{i-1}}{\sum_{i=1}^t s(s-i)(\frac{s+1}{s})^{i-1}}+(\frac{s+1}{s})^t.$$
Using Lemma~\ref{lem:id1} and \ref{lem:id2} to simply it further, we have
$$\frac{ f(O)}{ f(S)}\le\frac{2s(\frac{s+1}{s})^{2t}-2t(\frac{s+1}{s})^{t}-2s}{(2s-t)(\frac{s+1}{s})^{t}-2s}.$$ 
Let $x=(\frac{s+1}{s})^s$ and $r=\frac{t}{s}$, we study the continuous version of the above function
$$g(x,r)=\frac{2x^{2r}-2rx^{r}-2}{(2-r)x^{r}-2}.$$
Since $S$ is a local optimum with respect to the swapping of any single element and by the definition of $x, s$ and $t$, we have $2 \le t \le s$ and hence  
$2.25\le x \le e$ and $0< r \le 1$. Our goal then is to establish an upper
bound on $g(x,r)$ for $2.25 \le x \le e$ and $0 < r \le 1$. We will think
of $g(x,r)$ as implictly  defining $x$ as a function of $r$ at 
points where $g(x,r)$ can possibly take on a maximum value, namely when 
when $\frac{\partial g(x,r)}{\partial x} = 0$ and at the boundary 
points for $x$.
 
Note that since $x\ge 2.25$,
$$x>\left(\frac{2}{2-r}\right)^\frac{1}{r},$$
for all $0< r \le 1$. Therefore, we have $(2-r)x^{r}-2>0$ for given $x$ and $r$. It is easy to verify that function $g(x,r)$ is continuous and differentiable. For any fixed $r$, the function has two boundary points at $x=2.25$ and $x=e$, and taking partial derivative with respect to $x$, we have 
$$\frac{\partial g(x,r)}{\partial x}=\frac{2rx^{r-1}(x^r-1)[(2-r)x^r-(2+r)]}{[(2-r)x^r-2]^2}.$$ 
Therefore the only point where the partial derivative equals to zero is
$$x^*=(\frac{2+r}{2-r})^{\frac{1}{r}}.$$ 
Plugging this into the original expression for $g(x,r)$, we have
$$g(x^*,r)=\frac{2r^2+8}{(r-2)^2}.$$
The function $g(x^*,r)$ is monotonically increasing with respect to $r\in (0,1]$ and it has a maximum value of $10$ when $r=1$. 

Now it only remains to check the two boundary points $x=2.25$ and $x=e$. Note that these are fixed values. We now fix $x$, and take partial derivative with respect to $r$:
$$\frac{\partial g(x,r)}{\partial r}=\frac{2x^r(x^r-1)[(2\ln x-r\ln x+1)x^r-(2\ln x+r\ln x+1)]}{[(2-r)x^r-2]^2}.$$
Since $x^r>0$, $x^r-1>0$ and $[(2-r)x^r-2]^2>0$. If we can show that $$(2\ln x-r\ln x+1)x^r-(2\ln x+r\ln x+1)>0$$ then the function after fixing $x$ is monotonically increasing with respect to $r$. We use the Taylor expansion of $x^r$ at $x=0$. 
$$x^r>1+r\ln x+\frac{1}{2}r^2\ln^2x.$$
Therefore,
$$(2\ln x-r\ln x+1)x^r-(2\ln x+r\ln x+1)>r\ln x(2\ln x+r\ln^2x-\frac{1}{2}r^2\ln^2x-\frac{1}{2}r\ln x -1).$$
Note that we only need to check for the case when $x=e$ and $x=2.25$.
\be
\item Case $x=e$: $$2\ln x+r\ln^2x-\frac{1}{2}r^2\ln^2x-\frac{1}{2}r\ln x -1=1+\frac{1}{2}r-\frac{1}{2}r^2>0.$$
\item Case $x=2.25$: $$2\ln x+r\ln^2x-\frac{1}{2}r^2\ln^2x-\frac{1}{2}r\ln x -1>0.6+0.6r-0.5r-0.4r^2>0.$$
\ee
Therefore $(2\ln x-r\ln x+1)x^r-(2\ln x+r\ln x+1)>0$, and hence $\frac{\partial g(x,r)}{\partial r}>0$ for $x=2.25$ and $x=e$. Therefore the maximum is obtained when $r=1$.
Plug $r=1$ into the original formula, we have
$$g(x,1)=\frac{2x^2-2x-2}{x-2}.$$
Evaluate it for $x=e$ and $x=2.25$, we have
$g(e,1)=10.22$ and $g(2.25,1)=14.5$.
This completes the proof.

\qed
\end{proof}




\section{Conclusion and Open Problem}
\label{sec:conclusion}
Motivated by the max-sum diversification problem we are led
to study a generalization of monotone submodular functions that
we call weakly-submodular functions. This class includes the supermodular
max-sum dispersion problem. 

There are several open problemsi regarding the class of weakly submodular 
functions.  
First we would like to find other natural functions that are 
monotone and non-monotone weakly submodular. 
As we have shown, our class does for example contain some but not all
functions with small supermodular degree as well as some functions that
do not have small submodular degree. Indeed, weakly submodular
functions are incomparable with functions having small supermodular degree.  
Another obvious question is whether there is an analogue
of the marginal decreasing property that characterizes submodular
functions or at least analogues that would be a consequence of 
weak submodularity and would be useful in analyzing algorithms.  

In terms of computational problems 
regarding the optimization of monotone  weakly submodular functions many interesting 
questions remain. 
Similar to the maximization for an arbitrary matroid constraint using
local search, we would like 
to have a proof of the convergence of the greedy approximation bound
for the cardinality constraint. Another immediate open problem
is to close the gap between the upper and lower bounds we know for
approximating an arbitrary monotone weakly submodular function subject to
cardinality or matroid constraints. We note that although all of our individual examples in section \ref{sec:examples} can either be computed optimally or 
have better approximation
ratios than we can prove for the class of monotone weakly submodular 
functions, it does not follow that a sum of such functions can be 
computed with such good polynomial time approximations. 
It would also be of interest to
consider an approximation for maximizing a weakly submodular function 
subject to a knapsack constraint. Finally, are the efficient 
constant appropximation algorithms for maximizing  non monotone
weakly submodular functions.

\paragraph{Acknowledgment}
We thank Norman Huang for many comments and in particular for Proposition 
\ref{prop:complement}. 
This research is supported by the Natural Sciences and Engineering Research Council of Canada 
and the University of Toronto, Department of Computer Science.

\bibliographystyle{abbrv}
\nocite{*}
\bibliography{weakly-submodular}

\begin{thebibliography}{10}

\bibitem{Alon14}
N.~Alon.

\bibitem{Alonammw11}
N.~Alon, S.~Arora, R.~Manoikaran, D.~Moshkovitz, and O.~Weinstein.
\newblock Inapproximability of densest $\kappa$-subgraph from average case
  hardness.
\newblock {\em Unpublished manuscript}, 2011.

\bibitem{AndersenC09}
R.~Andersen and K.~Chellapilla.
\newblock Finding dense subgraphs with size bounds.
\newblock In {\em Algorithms and Models for the Web-Graph, 6th International
  Workshop, {WAW} 2009, Barcelona, Spain, February 12-13, 2009. Proceedings},
  pages 25--37, 2009.

\bibitem{BJKN10}
N.~Bansal, K.~Jain, A.~Kazeykina, and J.~Naor.
\newblock Approximation algorithms for diversified search ranking.
\newblock In {\em Proceedings of the 37th international colloquium conference
  on Automata, languages and programming: Part II}, ICALP'10, pages 273--284,
  Berlin, Heidelberg, 2010. Springer-Verlag.

\bibitem{BirnbaumG09}
B.~E. Birnbaum and K.~J. Goldman.
\newblock An improved analysis for a greedy remote-clique algorithm using
  factor-revealing lps.
\newblock {\em Algorithmica}, 55(1):42--59, 2009.

\bibitem{BJLY14}
A.~Borodin, A.~Jain, H.~C. Lee, and Y.~Ye.
\newblock Max-sum diversification, monotone submodular functions and dynamic
  updates (revised version).
\newblock {\em arXiv temporary ID 1040101}, August, 2014.

\bibitem{BLY12}
A.~Borodin, H.~C. Lee, and Y.~Ye.
\newblock Max-sum diversification, monotone submodular functions and dynamic
  updates.
\newblock In {\em PODS}, pages 155--166, 2012.

\bibitem{Bru69}
R.~A. Brualdi.
\newblock Comments on bases in dependence structures.
\newblock {\em Bulletin of the Australian Mathematical Society}, 1:161--167, 8
  1969.

\bibitem{BuchbinderFNS14}
N.~Buchbinder, M.~Feldman, J.~Naor, and R.~Schwartz.
\newblock Submodular maximization with cardinality constraints.
\newblock In {\em Proceedings of the Twenty-Fifth Annual {ACM-SIAM} Symposium
  on Discrete Algorithms, {SODA} 2014, Portland, Oregon, USA, January 5-7,
  2014}, pages 1433--1452, 2014.

\bibitem{BuchbinderFNS12}
N.~Buchbinder, M.~Feldman, J.~S. Naor, and R.~Schwartz.
\newblock A tight linear time (1/2)-approximation for unconstrained submodular
  maximization.
\newblock In {\em Proceedings of the 2012 IEEE 53rd Annual Symposium on
  Foundations of Computer Science}, FOCS '12, pages 649--658, Washington, DC,
  USA, 2012. IEEE Computer Society.

\bibitem{ConitzerSS05}
V.~Conitzer, T.~Sandholm, and P.~Santi.
\newblock Combinatorial auctions with \emph{k-}wise dependent valuations.
\newblock In {\em Proceedings, The Twentieth National Conference on Artificial
  Intelligence and the Seventeenth Innovative Applications of Artificial
  Intelligence Conference, July 9-13, 2005, Pittsburgh, Pennsylvania, {USA}},
  pages 248--254, 2005.

\bibitem{Feige98}
U.~Feige.
\newblock A threshold of ln n for approximating set cover.
\newblock {\em J. ACM}, 45(4):634--652, July 1998.

\bibitem{FeigeFIILS14}
U.~Feige, M.~Feldman, N.~Immorlica, R.~Izsak, B.~Lucier, and V.~Syrgkanis.
\newblock A unifying hierarchy of valuations with complements and substitutes.
\newblock {\em Electronic Colloquium on Computational Complexity {(ECCC)}},
  21:103, 2014.

\bibitem{FeigeI13}
U.~Feige and R.~Izsak.
\newblock Welfare maximization and the supermodular degree.
\newblock {\em Electronic Colloquium on Computational Complexity {(ECCC)}},
  20:95, 2013.

\bibitem{FeigePK01}
U.~Feige, G.~Kortsarz, and D.~Peleg.
\newblock The dense \emph{k}-subgraph problem.
\newblock {\em Algorithmica}, 29(3):410--421, 2001.

\bibitem{FeigeMV2007}
U.~Feige, V.~S. Mirrokni, and J.~Vondrak.
\newblock Maximizing non-monotone submodular functions.
\newblock In {\em Proceedings of the 48th Annual IEEE Symposium on Foundations
  of Computer Science}, FOCS '07, pages 461--471, Washington, DC, USA, 2007.
  IEEE Computer Society.

\bibitem{FeldmanI14}
M.~Feldman and R.~Izsak.
\newblock Constrained monotone function maximization and the supermodular
  degree.
\newblock In {\em Approximation, Randomization, and Combinatorial Optimization.
  Algorithms and Techniques, {APPROX/RANDOM} 2014, September 4-6, 2014,
  Barcelona, Spain}, pages 160--175, 2014.

\bibitem{GS09}
S.~Gollapudi and A.~Sharma.
\newblock An axiomatic approach for result diversification.
\newblock In {\em Proceedings of the 18th international conference on World
  wide web}, WWW '09, pages 381--390, New York, NY, USA, 2009. ACM.

\bibitem{Jenkyns1976}
T.~Jenkyns.
\newblock The efficacy of the `greedy' algorithm.
\newblock In {\em Proc. of 7th Southeastern Conf. on Combinatorics, Graph
  Theory and Computing}, pages 341--350, 1976.

\bibitem{KKT03}
D.~Kempe, J.~Kleinberg, and E.~Tardos.
\newblock Maximizing the spread of influence through a social network.
\newblock In {\em Proceedings of the ninth ACM SIGKDD international conference
  on Knowledge discovery and data mining}, KDD '03, pages 137--146, New York,
  NY, USA, 2003. ACM.

\bibitem{KhullerS09}
S.~Khuller and B.~Saha.
\newblock On finding dense subgraphs.
\newblock In {\em Automata, Languages and Programming, 36th International
  Colloquium, {ICALP} 2009, Rhodes, Greece, July 5-12, 2009, Proceedings, Part
  {I}}, pages 597--608, 2009.

\bibitem{KPR09}
J.~Kleinberg, C.~Papadimitriou, and P.~Raghavan.
\newblock Segmentation problems.
\newblock In {\em Proceedings of the thirtieth annual ACM symposium on Theory
  of computing}, STOC '98, pages 473--482, New York, NY, USA, 1998. ACM.

\bibitem{KPR98}
M.~R. Korupolu, C.~Plaxton, and R.~Rajaraman.
\newblock Analysis of a local search heuristic for facility location problems.
\newblock {\em Journal of Algorithms}, 37(1):146--188, 2000.

\bibitem{LB11}
H.~Lin and J.~Bilmes.
\newblock A class of submodular functions for document summarization.
\newblock In {\em Proceedings of the 49th Annual Meeting of the Association for
  Computational Linguistics: Human Language Technologies - Volume 1}, HLT '11,
  pages 510--520, Stroudsburg, PA, USA, 2011. Association for Computational
  Linguistics.

\bibitem{Mestre06}
J.~Mestre.
\newblock Greedy in approximation algorithms.
\newblock In {\em Algorithms - {ESA} 2006, 14th Annual European Symposium,
  Zurich, Switzerland, September 11-13, 2006, Proceedings}, pages 528--539,
  2006.

\bibitem{NWF78}
G.~Nemhauser, L.~Wolsey, and M.~Fisher.
\newblock An analysis of the approximations for maximizing submodular set
  functions.
\newblock {\em Mathematical Programming}, 1978.

\bibitem{Wild08}
M.~Wild.
\newblock Weakly submoudlar rank functions, supermatoids, and the flat lattice
  of a distributive supermatroid.
\newblock {\em Discrete Mathematics}, 308(7):999--1017, 2008.

\bibitem{Ye13}
Y.~Ye.
\newblock General context amenable to greedy and greedy-like algorithms.
\newblock {\em PHD Thesis, University of Toronto}, March 2013.

\end{thebibliography}

\end{document}